%% file: paper-twitter.tex
\documentclass{sig-alternate-2013}
\usepackage{color}
\usepackage{graphicx}
\usepackage{multirow}
\usepackage{times}
\usepackage{url}
\usepackage{enumitem}
\usepackage{algpseudocode, algorithm}
\usepackage{xspace}
\usepackage[normalem]{ulem}

\newfont{\mycrnotice}{ptmr8t at 7pt}
\newfont{\myconfname}{ptmri8t at 7pt}
\let\crnotice\mycrnotice%
\let\confname\myconfname%

\permission{Permission to make digital or hard copies of all or part of this work for personal or classroom use is granted without fee provided that copies are not made or distributed for profit or commercial advantage and that copies bear this notice and the full citation on the first page. Copyrights for components of this work owned by others than ACM must be honored. Abstracting with credit is permitted. To copy otherwise, or republish, to post on servers or to redistribute to lists, requires prior specific permission and/or a fee. Request permissions from Permissions@acm.org.}
\conferenceinfo{KDD'15,}{August 10-13, 2015, Sydney, NSW, Australia.}
\copyrightetc{\copyright~2015 ACM. ISBN \the\acmcopyr}
\crdata{978-1-4503-3664-2/15/08\ ...\$15.00.\\
DOI: http://dx.doi.org/10.1145/2783258.2783401}

\def\Title {SEISMIC: A Self-Exciting Point Process Model \\ for Predicting Tweet Popularity}

\input{macros}

\newcommand{\hide}[1]{}
\newcommand{\xhdr}[1]{\vspace{1.7mm}\noindent{{\bf #1.}}}
\newcommand{\eg}{\emph{e.g.}}
\newcommand{\ie}{\emph{i.e.}}

\newcommand{\sem}{{\sc{Seismic}}\xspace} 

\graphicspath{{./}}

\begin{document}

\title{\Title}

\numberofauthors{5} 
%
\author{
\alignauthor
Qingyuan Zhao\\
      \affaddr{Stanford University}\\
      \email{qyzhao@stanford.edu}
\alignauthor
Murat A. Erdogdu\\
      \affaddr{Stanford University}\\
      \email{erdogdu@stanford.edu}
\alignauthor
Hera Y. He\\
      \affaddr{Stanford University}\\
      \email{yhe1@stanford.edu}
\and  
\alignauthor
Anand Rajaraman\\
      \affaddr{Stanford University}\\
      \email{anand@cs.stanford.edu}
\alignauthor
Jure Leskovec\\
      \affaddr{Stanford University}\\
      \email{jure@cs.stanford.edu}
}

\maketitle
\begin{abstract}
\input{abstract}
\end{abstract}

\vspace{2mm}
\noindent {\bf Categories and Subject Descriptors:} H.2.8 {\bf
[Database Management]}: Database applications---{\it Data mining}

\noindent {\bf General Terms:} Algorithms; Experimentation.


\section{Introduction}
\label{sec:introduction}
\input{Intro}

\section{Related Work}
\label{sec:related-work}
\input{relatedWork}

\section{\!\!\! Modeling Information Cascades}
\label{sec::genFrame}
\input{GenFramework}

\section{\!\!\!\!\! Predicting Information Cascades}
\label{sec::theory}
\input{Theory}

\section{Experiments}
\label{sec:experiments}
\input{Experiments}


\section{Conclusion and Future Work}
\label{sec:discuss}
\input{Discussion}

\section*{Data and Software}
The \sem software and the dataset we use in Section
\ref{sec:experiments} can be found in
\url{http://snap.stanford.edu/seismic/}. The R package of our
algorithm is also available on \url{http://cran.r-project.org/web/packages/seismic}.

\section*{Acknowledgements}
The authors would like to thank David O. Siegmund for his constructive
suggestions and Austin Benson, Bhaswar B. Bhattacharya, Joshua Loftus
for their helpful comments. This research has been supported in part by NSF
IIS-1016909,              
CNS-1010921,              
IIS-1149837,              
IIS-1159679,              
ARO MURI,                 
DARPA SMISC, DARPA SIMPLEX,
Stanford Data Science Initiative,
Boeing,
Facebook,
Volkswagen,
and Yahoo.

\bibliographystyle{abbrv} 
\bibliography{ref}

\end{document}

%% file: macros.tex
\newtheorem{theorem}{Theorem}[section]

\newtheorem{proposition}[theorem]{Proposition}

\newcommand{\commentout}[1]{}
\def\S{\mathcal{S}}

\def\p{\hat{p}}

\def\la{\lambda}

\def\<{\langle}
\def\>{\rangle}
\def\E{{\mathbb E}}
\def\P{{\mathbb P}}

\def\F{{\mathcal F}}

\newcommand{\eqn}[1]{\begin{align*}#1\end{align*}}

%% file: abstract.tex

Social networking websites 
allow users to create and share content. Big information
cascades of post resharing can form as users of these sites reshare
others' posts with their friends and followers. 
One of the central challenges in understanding such cascading
behaviors is in forecasting information outbreaks, where a single post becomes widely popular by being reshared by many users.

In this paper, we focus on predicting the final number of reshares of
a given post. We build on the theory of self-exciting point processes
to develop a statistical model that allows us to make accurate predictions.
Our model requires no training or expensive feature engineering. It
results in a simple and efficiently computable formula that allows us
to answer questions, in real-time, such as:
Given a post's resharing history so far, what is our current estimate
of its final number of reshares? Is the post resharing cascade past the initial stage of explosive growth?
And, which posts will be the most reshared in the future?

We validate our model using one month of complete Twitter data and
demonstrate a strong improvement in predictive accuracy over existing
approaches. Our model gives only 15\% relative error in predicting
final size of an average information cascade after observing it for just one hour.

\hide{
We consider the problem of forecasting information outbreaks on social
networks. Many social networks allow their users to share others' post
with their own friends or followers, which may form many large
information cascades over the network. In this paper, we propose a
self-exciting statistical model to analyze the behavior of such information
cascades. Our model can be used in real time
to: 1. answer when the cascade is large enough to be predictable;
2. classify the cascade into breakout or non-breakout; 3. predict its
final influence. We use a large Twitter data set to validate the
effectiveness of our model, which indeed shows a clear improvement over
several benchmarks.
} 

%% file: Intro.tex
Online social networking services, such as Facebook, Youtube, and Twitter, allow their users to post and share content in the form of posts, images, and videos~\cite{dow2013anatomy,hong2011predicting,nowell08letter,suh2010want}. As a user is exposed to posts of others she follows, the user may in turn reshare a post with her own followers, who may further reshare it with their respective sets of followers. This way large information cascades of post resharing spread through the network.

A fundamental question in modeling information cascades is to predict their future evolution.
Arguably the most direct way to formulate this question is to consider predicting the final size of an information cascade. That is, to predict how many reshares a given post will ultimately receive.

Predicting \hide{the attention or} the ultimate popularity of a post is
important for content ranking and aggregation. For instance, Twitter
is overflowing with posts and users have a hard time keeping up with
all of them. Thus, much of the content gets missed and eventually
lost. Accurate prediction would allow Twitter to rank content better,
discover trending posts faster, and improve 
its content-delivery networks.
Moreover, predicting information cascades allows us to gain fundamental insights into predictability of collective behaviors where uncoordinated actions of many individuals lead to spontaneous outcomes, for example, large information outbreaks.

Most research on predicting information cascades involves extracting
an exhaustive set of features describing the past evolution of a
cascade and then using these features in a simple machine learning
classifier to make a prediction about future growth~\cite{bandari2012pulse,cheng2014can,hong2011predicting,kupavski2012cikm,petrovic2011rt,suh2010want}. However,
feature extraction can be expensive and cumbersome, and one is never
sure if more effective features could be extracted. The question
remains how to design a simple and principled
bottom-up model of cascading behavior. The challenge lies in defining
a model for an individual's behavior and then aggregating the effects
of the individuals in order to make an accurate global prediction.

\begin{figure}[!t]
  \centering
  \includegraphics[width = \columnwidth]{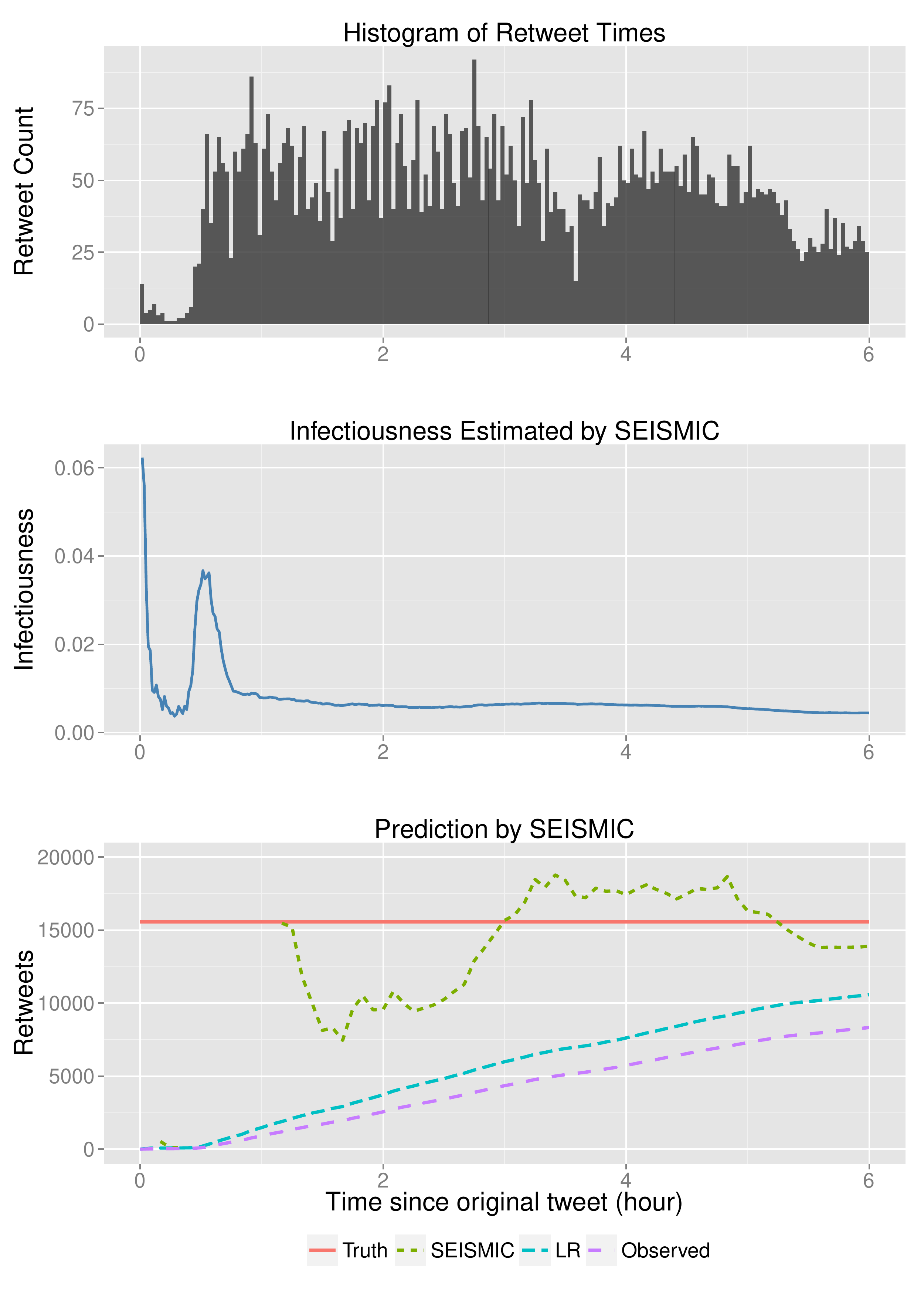}
  \caption{First 6 hours of retweeting activity of a popular
    tweet~\cite{gaddafi-bieber} (top). The controversial tweet is
    about the fresh death of dictator Muammar Gaddafi and mentions singer Justin Bieber.
  Interestingly, the car manufacturer Chevrolet Twitter account
  inappropriately retweeted the tweet about 30 minutes after the original
  tweet, which possibly lead to tweet's   sustained popularity.
  Tweet infectiousness against time as estimated by \sem (middle).
  Predictions of the tweet's final retweet count (denoted as ``Truth'') as a function of time
  (bottom). We compare \sem with time series linear regression (LR), ``Observed'' plots the cumulative number of observed retweets by a given time. Notice \sem quickly finds an
  accurate estimate of the tweet's final retweet count.
  }
  \label{fig:breakout}
\end{figure}

\xhdr{Present work}
Here we focus on predicting the final size of an information cascade
spreading through a network. We develop a statistical model based on the theory of {\em self-exciting point
  processes}. A point process indexed by time is called a {\em counting process}
when it counts the number of instances (reshares, in our case) over time. In contrast to homogeneous
Poisson processes which assume constant intensity over time, self-exciting processes assume that all the previous
instances (\ie, reshares) influence the future evolution of the
process. Self-exciting point processes are frequently used to model ``rich get richer'' phenomena~\cite{Matsubara2012rise,mohler2011,yang2013mixture,zhou2013learning}. They are ideal for modeling information cascades in networks because every new reshare of a post not only increases its cumulative reshare count by one, but also exposes new followers who may further reshare the post.


We develop \sem (\emph{Self-Exciting Model of Information Cascades})
for predicting the total number of reshares of a given post. In our
model, each post is fully characterized by its {\em infectiousness} which measures the reshare probability. We allow the infectiousness
to vary freely over time in agreement with the observation that the
infectiousness can drop as the content gets stale (see~Figure~\ref{fig:breakout}). Moreover, our model is able to identify at
each time point whether the cascade is in the \emph{supercritical} or
\emph{subcritical} state, based on whether its infectiousness is above
or below a critical threshold. A cascade in the supercritical state is
going through an ``explosion'' period and its final size cannot be
predicted accurately at the current time. On the contrary, a cascade
is tractable if it is in subcritical state. In this case, we are able to predict its
ultimate popularity accurately by modeling the future cascading behavior by a
Galton-Watson tree.


Our \sem approach makes several contributions:
\begin{itemize}[nolistsep]
\item {\bf Generative model:} \sem imposes no parametric
  assumptions and requires no expensive feature engineering.
  Moreover, as complete social network structure may be hard to obtain,
  \sem assumes minimal knowledge of the network:
  The only required input is the time history of reshares and the degrees of the
  resharing nodes.
\item {\bf Scalable computation:} 
  Making a prediction using \sem only requires computational time linear in the
  number of observed reshares. Since predictions for individual posts
  can be made independently, our algorithm can also be easily
  parallelized. 
\item {\bf Ease of interpretation:} For an individual cascade, the model synthesizes all its past
  history into a single infectiousness parameter. This infectiousness
  parameter holds a clear meaning, and can serve as input to other applications.
\end{itemize}

We evaluate \sem on one month of complete Twitter data, where users post tweets which others can then reshare by retweeting them. 
We demonstrate that \sem is able to predict the final retweet count of
a given tweet with 30\% better accuracy than the state-of-the-art
approaches (\eg, \cite{gao2015modeling}). For reasonably popular tweets, our model achieves 15\%
relative error in predicting the final retweet count after observing
the tweet for 1 hour, and 25\% error after observing the tweet for just 10
minutes. Moreover, we also demonstrate how \sem is able to
identify tweets that will go ``viral'' and be among the most popular
tweets in the future.
By maintaining a dynamic list of 500 tweets over time, we are able to
identify 78 of the 100 most reshared tweets and 281 of the 500
most reshared tweets in just 10 minutes after they are posted.

The rest of the paper is organized as follows: Section \ref{sec:related-work} surveys the related work. Section~\ref{sec::genFrame} describes \sem, and Section~\ref{sec::theory} shows how the model can be used to predict the final size of an information cascade. We evaluate our method and compare its performance with a number of baselines as well as state-of-the-art approaches in Section~\ref{sec:experiments}. Last, in Section \ref{sec:discuss}, we conclude and discuss future research directions.



%% file: relatedWork.tex

The study of information cascades is a rich and active
field~\cite{rogers95diffusion}. Recent models for predicting size
of information cascades are generally characterized by two types of approaches, feature based methods and point process based methods.

Feature based methods first extract an exhaustive list of potentially relevant features, including content features, original poster features, network structural features, and temporal features \cite{cheng2014can}.
Then different learning algorithms are applied, such as simple
regression models~\cite{Agarwal2009,cheng2014can}, probabilistic
collaborative filtering~\cite{zaman2010predicting}, regression
trees~\cite{bakshy2011everyone}, content-based models~\cite{naveed2011bad}, and passive-aggressive algorithms~\cite{petrovic2011rt}.
There are several issues with such approaches: laborious feature engineering and extensive
training are crucial for their success, and the performance is highly
sensitive to the quality of the
features~\cite{bandari2012pulse,suh2010want}. Such approaches also
have limited applicability because they cannot be used in real-time
online settings---given the massive amount of posts being produced every second, it is practically impossible to extract all the necessary features for every post and then apply complicated prediction rules.
In contrast, \sem requires no feature engineering and results in an
efficiently computable formula that allows it to predict the final
popularity of millions of posts as they are spreading through the
network.

The second type of approach is based on point processes, which directly models the formation of an information cascade in a network.
Such models were mostly developed for the complementary problem of
network inference, where one observes a number of information cascades
and tries to infer the structure of the underlying network over which
the cascades propagated~\cite{daneshmand14netrate, du2012learning, manuel13pathways,rodriguez2013modeling, manuel13dynamic, hunter2011dynamic, yang2013mixture,zhou2013learning}.
These methods have been successfully applied to study the spread of memes on
the web~\cite{du2012learning,rodriguez2013modeling,yang2011,yang2013mixture} as well as hashtags on Twitter~\cite{zhou2013learning}. In contrast, our goal is not to infer the network but to predict the ultimate size of a cascade in an observed network.

A major distinction between our model and existing methods based on
Hawkes processes ({\it e.g.}, \cite{Matsubara2012rise,mohler2011,yang2013mixture,Zaman2014,zhou2013learning}) is that we assume the process intensity $\lambda_t$ depends on another
stochastic process $p_t$, the post infectiousness. In other words, we
allow the infectiousness to change over time. Moreover, some of these
methods \cite{Zaman2014} rely on computationally expensive Bayesian inference, while our method has linear time complexity.
Another recently proposed related work is \cite{gao2015modeling},
which also takes the point process approach and directly aims to
predict tweet popularity. However, their method makes restrictive
parametric assumptions and does not consider the network structure,
which limits its predictive ability. We compare \sem with
\cite{gao2015modeling} in Section~\ref{sec:experiments} and demonstrate a 30\% improvement.

%% file: GenFramework.tex

In this section, 
we describe \sem and discuss how it can be used for:

\begin{enumerate}[nolistsep]
  \item Estimating the spreading rate of a given information cascade, which we quantify by the post's infectiousness.
  \item Determining whether the cascade is in supercritical (explosive) or subcritical (dying out) state.
  \item Predicting the final size of an information cascade, which is measured by the ultimate number of reshares received by the post that started the cascade.
\end{enumerate}

Important quantities in our model are the total number of
reshares $R_t$ of a given post up to time $t$ and the
cascade speed of spreading $\lambda_t$. In our model, $\lambda_t$ is determined by the post
infectiousness $p_t$ and human reaction time. Our goal is to predict $R_\infty$, the final number of reshares.

Another important quantity in our model is the {\em memory kernel}
$\phi(s)$, which quantifies the delay between a post arriving to a user's feed and the user resharing it. Intuitively, infectiousness defines the probability that a given user will reshare a given post, and the memory kernel models user's reaction time. By combining the two we can then accurately model the speed at which the post will spread through the network.
Table~\ref{tab:notations} summarizes the notation.

\begin{table}[!t]
  \centering
  \begin{tabular}{l | l}
    Symbol & Description \\
    \hline
    $w$ & Post/information cascade \\
    $p_t$ & Infectiousness of $w$ at time $t$ (Section \ref{sec:spread-speed}) \\
    $\phi(s)$ & Memory kernel (Section \ref{sec:human-reaction-time}) \\
    $i$ & Node that contributed $i^\textit{th}$ reshare. \\
         &$i=0$ corresponds to the originator of the post. \\
    $t_i$ & Time of the $i^\textit{th}$ reshare relative to the original
    post.\\
    $n_i$ & Out-Degree of the $i^\textit{th}$ node \\
    $R_t$ & Cumulative popularity by time $t$: $|\{i > 0; t_i \le t\}| $ \\
    $R_{\infty}$ & Final popularity (final number of reshares): $|\{i
    > 0\}|$ \\
    $N_t$ & Cumulative degree of resharers by time $t$: $\sum_{i: t_i\le t} n_i$ \\
    $N_t^e$ & Effective cumulative degree of resharers by time $t$:\\
            & $N_t^e = \sum_{i=0}^{R_t}  n_i \int_{t_{i}}^{t} \phi(s-t_i)ds $\\
    $\lambda_t$ & Intensity of cumulative popularity $R_t$ \\
    $\hat{p}_t$ & Model's estimate of infectiousness $p_t$ at time $t$ \\
    $\hat{R}_{\infty}(t)$ & Model's estimate at time $t$ of final popularity $R_{\infty}$ \\
  \end{tabular}
  \caption{Table of symbols.}
  \label{tab:notations}
\end{table}


\subsection{Human reaction time}
\label{sec:human-reaction-time}
In order to predict the cascade size, we need to know how long it takes for a person to reshare a post. Knowing the delay allows us to accurately model the speed of a cascade spreading through the network. We consider that the time $s$ between the arrival of a post in a users'
timeline and a reshare of the post by the user is distributed with
density $\phi(s)$. The probability density $\phi(s)$ is also called a memory kernel because it measures a physical/social system's memory of stimuli~\cite{Crane2008}.

The distribution of human response time $\phi(s)$ has been shown to be heavy-tailed in social networks~\cite{Barabasi2005}. Usually
the tail of $\phi(s)$ is assumed to follow a power-law with exponent
between $1$ and $2$ or a log-normal distribution~\cite{Crane2008,Zaman2014}.
However, due to the rapid nature of information sharing on Twitter, it
is also natural to expect many instant reaction times. In fact, our exploratory data analysis in Section
\ref{sec:parameter-estimation} confirms that in Twitter, $\phi(s)$ is
approximately constant for the first 5 minutes and then followed by a
power-law decay. Different social networks may have different
distributions of human reaction times. However, $\phi(s)$ only needs
to be estimated once per network and thus we can safely
assume it is given. We describe a detailed estimation procedure of $\phi(s)$ in Section~\ref{sec:parameter-estimation}.

\subsection{Post infectiousness}
\label{sec:spread-speed}

The second component of our model is the post infectiousness. We assume each
post $w$ is associated with a time dependent, intrinsic
infectiousness parameter $p_t(w)$. In other words, $p_t(w)$ models how
likely the post $w$ is to be reshared at time $t$.
Infectiousness of a post may depend on a combination of factors,
including but not limited to the quality of the post's content, the social network
structure, the current local time, and the geographical location. Instead of assuming a parametric form of $p_t$, we model it
flexibly in a nonparametric way which implicitly accounts for all
these factors. 

Most existing methods studying self-exciting point processes assume
$p_t$ to be fixed over time. Consequently, an important concept is the
{\em criticality} of the process $R_t$. In a self-exciting point process
with constant infectiousness $p_t \equiv p$, there exists a phase transition
phenomenon at certain critical threshold $p^{*}$ such that~\cite{durrett2010probability}:
\begin{enumerate} [nolistsep]
  \item If $p > p^{*}$, then $R_t \to \infty$ as $t \to \infty$ almost surely
  and exponentially fast. This is called the \emph{supercritical} regime.
  \item If $p < p^{*}$, then $\sup_t R_t < \infty$ almost surely. This is
  called the \emph{subcritical} regime.
\end{enumerate}

In reality, $R_t$ is always bounded due to the finite size of the network.
Thus, no supercritical
cascades can exist if $p_t$ is assumed to be a constant. This is
inadequate to model highly contagious tweets and our assumption of
non-constant infectiousness solve this problem as well. Furthermore, as the post gets older the
information becomes outdated and its spreading power (infectiousness) may decrease.
This effect may also be observed as the post spreads
further away from the original poster~\cite{Zaman2014}. Alternatively,
resharing by a highly influential user may increase the
post's infectiousness. Thus, rather than assuming a common evolutionary pattern of $p_t$ for all the tweets, we only assume it
varies smoothly over time and use non-parametric methods to estimate
$p_t$ for each tweet.

\subsection{The SEISMIC model}
\label{sec:self-exiciting-model}

We combine human reaction times and post
infectiousness to derive \sem. In order to link $p_t$ to the post resharing
process $R_t$, we model $R_t$ as a \emph{doubly stochastic
  self-exciting point process}. This is an extension to the standard self-exciting
point process (also called the Hawkes process~\cite{hawkes1971}) which
was initially used to model earthquakes~\cite{ogata1988}.

We first define the intensity $\la_t$ of $R_t$, which simply measures the rate of obtaining an additional reshare at time $t$. More formally:
\begin{equation*}
  \label{eq:1}
\la_t =\lim_{\Delta \downarrow 0} \frac{\P\left ( R_{t+\Delta } - R_t
    = 1 \right )}{\Delta }.
\end{equation*}

In \sem, the intensity $\la_t$ at time $t$ is determined by
infectiousness $p_t$, reshare times $t_i$, node degrees  $n_i$, and
human reaction time distribution $\phi(s)$. The exact relationship
described in Eq.~\eqref{eq:model} is inspired by the theory of Hawkes processes~\cite{hawkes1971}:
\begin{equation}
  \label{eq:model}
  \begin{aligned}
  \lambda_t &= p_t \cdot \sum_{t_i \le t,~i \ge 0} n_i \phi(t-t_i), \quad t \ge t_0.
  \end{aligned}
\end{equation}

Intuitively, $\sum_{t_i \le t,~i \ge 0} n_i \phi(t-t_i)$ is the
intensity of the arrival of newly exposed users at time $t$ and its
product with the resharing probability $p_t$ gives the
intensity of reshares at time $t$.

Note that the above point process is called \emph{self-exciting}
because each previous observation $i$ such that $t_i \le t$
contributes to the intensity $\lambda_t$, or equivalently, each
observation increases the intensity in the future. It is further
\emph{doubly stochastic} (or a \emph{Cox process}) because the infectiousness $p_t$ is itself a stochastic process.

Additionally, we assume node degrees $\{n_i\}$ are independent and
identically distributed with mean $n^{*}$. Mean degree $n^{*}$ is
related to the critical threshold $p^{*}$ which is already discussed in Section \ref{sec:spread-speed}. The critical infectiousness threshold takes value $p^{*} = 1 / n^{*}$. We give the proof of this fact in Proposition~\ref{prop:expectation}.



%% file: Theory.tex


In this section we describe how to perform statistical inference for the self-exciting
model of information cascades introduced in the previous section.
Specifically, we discuss how \sem estimates the
infectiousness parameter $p_t$ and then predicts the ultimate size of
the cascade $R_{\infty}$.

Throughout this section, we make a technical assumption that
the followers of all the resharers are disjoint, so we can use a tree
structure to describe the information diffusion (Figure
\ref{fig:descendants}).
The conclusions made in this section remain valid even if resharers are not disjoint.
In this case, we can replace the node degree $n_i$ with the total
number of newly exposed neighbors of node $i$ (the followers of
$i$-th resharer who do not follow the first $i-1$ resharers).

\subsection{Estimating post infectiousness}
\label{sec:estimation-infectiousness}

We first define the {\em sample-function density}, which plays a central
role in estimating self-exciting point processes~\cite{snyder2011}.
Let's denote $\mathcal{F}_t = \sigma \left( \{ (n_i, t_i)
  \}_{i=0}^{R_t} \right)$ as the $\sigma$-algebra generated by all the
information available by time $t$: the times $t_i$ of all the reshares up to time $t$
and the number of followers (\ie, node degree) $n_i$ of the $i$-th
user to reshare.
Sample-function density is defined
as the joint probability of the number of reshares in the time interval
$[t_0, t)$ and the density of their occurrence times.

To motivate our estimator of $p_t$, we first consider the case where the
infectiousness parameter remains constant over time, \ie, $p_t \equiv p$.
Later we will relax this assumption and allow $p_t$ to vary over time.

In \sem, 
the sample-function density can be expressed using the intensity
$\lambda_t$ as \cite[Thm. 6.2.2]{snyder2011}
\begin{equation}
  \label{eq:likelihood}
  \begin{aligned}
  \P(R_t = r, t_1, \ldots, t_r) =
\prod_{i=1}^{R_t} \lambda_{t_i} \cdot \exp \left \{
  - \int_{t_0}^{t} \lambda_s ds \right \}. \\
    \end{aligned}
\end{equation}
By taking derivative of the log of Eq.~\eqref{eq:likelihood} and combining it with Eq.~\eqref{eq:model}, we obtain the maximum likelihood estimate (MLE) of $p_t$:
\begin{equation}
  \label{eq:mle}
  \p_t \ =\ \frac{R_t}{\sum_{i=0}^{R_t} n_i \int_{t_{i}}^{t} \phi(s-t_i)
     ds }
\end{equation}

The above equation forms the basis of \sem as it allows us to
estimate the infectiousness $\p_t$ at any given time $t$. Moreover, a
confidence interval of $p_t$ can also be obtained~\cite{snyder2011}.


The denominator in Eq. \eqref{eq:mle}, denoted as $N_t^e$ hereafter,
can be interpreted as the accumulative
``effective'' number of exposed users to the post. The numerator $R_t$ is the current number
of reshares of the post. To shed more light on our
estimator, we take $t \rightarrow \infty$, which leads to:
\begin{equation}
\label{eq:p-infty}
\p(\infty) = \frac{1}{\frac{1}{R_{\infty}}\sum_{j=0}^{R_{\infty}}
  n_j } \approx \frac{1}{n_{*}}
\end{equation}
Thus, by assuming the infectiousness $p_t$ to be a constant over
time, one would essentially assume that most posts have the same
infectiousness $1/n_{*}$. However, such assumption is unrealistic as it
cannot explain the bursty and volatile dynamics information cascades (\eg, Figure \ref{fig:breakout}). 

This undesirable consequence of assuming constant $p_t$ is another
motivation for allowing $p_t$ to vary over time. To estimate $p_t$ in this case, we smooth the MLE in Eq. \eqref{eq:mle} by only
using observations close to time $t$ to estimate $p_t$. In particular,
we rely on a sequence of one-sided kernels $K_t(s)$, $s >
0$, indexed by time $t$. We use these kernels to weight the reshares
and the weighted estimate of $p_t$ is given by
\begin{equation}
  \label{eq:mle-kernel}
\begin{aligned}
  \p_t &= \frac{\int_{t_0}^t K_t(t - s) d R_s}{\int_{t_0}^t K_t(t - s)
    d N^e_s } \\
&= \frac{\sum_{i=1}^{R_t} K_t(t - t_i)}{\sum_{i=0}^{R_t}
    n_i\int_{t_i}^t K_t(t - s) \phi(s - t_i) ds}\, .
\end{aligned}
\end{equation}

Notice that when $K_t(s) \equiv 1$ the estimator reduces to the MLE we derived in Eq.~\eqref{eq:mle}. In \sem we use a triangular kernel with growing window size $t/2$ as weighting kernel $K_t(s)$:
\begin{equation}
  \label{eq:weighting-kernel}
  K_t(s) = \max \left\{1 - \frac{2s}{t}, 0\right\}, \ \ \ \ s >0.
\end{equation}

We chose the triangular kernel because it has properties important for our application.
First, the kernel discards all posts
that are older than $t/2$. In particular, it quickly discards the
unstable and potentially explosive period at the beginning, which if included, would
introduce an upward bias to $p_t$.
Second, the kernel takes into account posts in
a larger window size as time $t$ increases. According to
our experiments, the growing window size helps to stabilize $\hat{p}(t)$
compared to a fixed window size.
Third, for reshares
within the window, the kernel up-weights the most recent posts and
gradually down-weights older posts. This keeps our estimator
$\hat{p}(t)$ closer to the ever-changing true $p_t$.
And last, as $K_t(s)$ is piece-wise linear, the integral $\int K_t(t
  - s)\phi(s - t_i) ds$ has a closed form for
many different functions $\phi(s)$ including the one we use for
Twitter in our experiments, see Section \ref{sec:experiments}.

\subsection{Predicting final popularity}
\label{sec:predict-influence}

Having described the procedure for inferring the post infectiousness,
we now need to account for the network structure in order to predict
how far the post is going to spread across the network.

For simplicity, let us assume the post is first posted at time $0$,
\ie, $t_0 = 0$. Consider we have observed the post for $t$ time units
and our goal now is to predict the post's final reshare count, $R_\infty$, based on the information we have observed so far, $\mathcal{F}_t$.

The following proposition shows how to compute the expected final
reshare count of a post. The main idea is to model an information
cascade spreading over the network with a branching process that
counts the reshare number of a post, as illustrated in Figure \ref{fig:descendants}.
Predictor for $R_\infty$ used by \sem can be stated as follows:

\begin{proposition}\label{prop:expectation}
Assume the (out-)degrees in the network are i.i.d. with expectation
$n_*$ and the infectiousness parameter $p_s$ is a constant $p$ for $s
\ge t$. Then, we have
\begin{equation}
\E[R_{\infty}|
\ \mathcal{F}_t]
=\left\{
  \begin{aligned}
&R_t + \frac{p (N_t - N^e_t)}{1 - p n_{*}},\!\!&\mathrm{if~} p <
\frac{1}{n_{*}}, \\
&\infty, ~&\mathrm{if~} p \geq \frac{1}{n_{*}}.
  \end{aligned} \right.
\label{eq:predict}
\end{equation}
\end{proposition}

\begin{proof}
First, we consider the case where $p < 1/n_*$.
We define a sequence of random variables
$\{Z_{1},Z_{2},Z_{3},\ldots\}$ that models the future information
diffusion tree, as illustrated in Figure \ref{fig:descendants}. In
this tree, $Z_k$ denotes the number of reshares made by the $k^\textit{th}$ generation descendants (counting from generation $R_t$ onward). 
Thus, the 1$^\textrm{st}$ generation descendants $Z_1$ refers to the
number of new reshares generated by the posts created before time $t$,
the 2$^\textrm{nd}$ generation descendants $Z_2$ refers to the reshares of the posts of the 1$^\textrm{st}$ descendants, and so on.
Notice that the summation over the $Z_k$'s gives the post's final
reshare count $R_{\infty} = R_t + \sum_{k=1}^{\infty} Z_k$. In the
following we use descendants $Z_k$ only for deriving
Eq. \eqref{eq:predict} and emphasize that our final estimator does not
require explicit network structure information.

\begin{figure}[t]
\centering
\includegraphics[width=0.9\columnwidth]{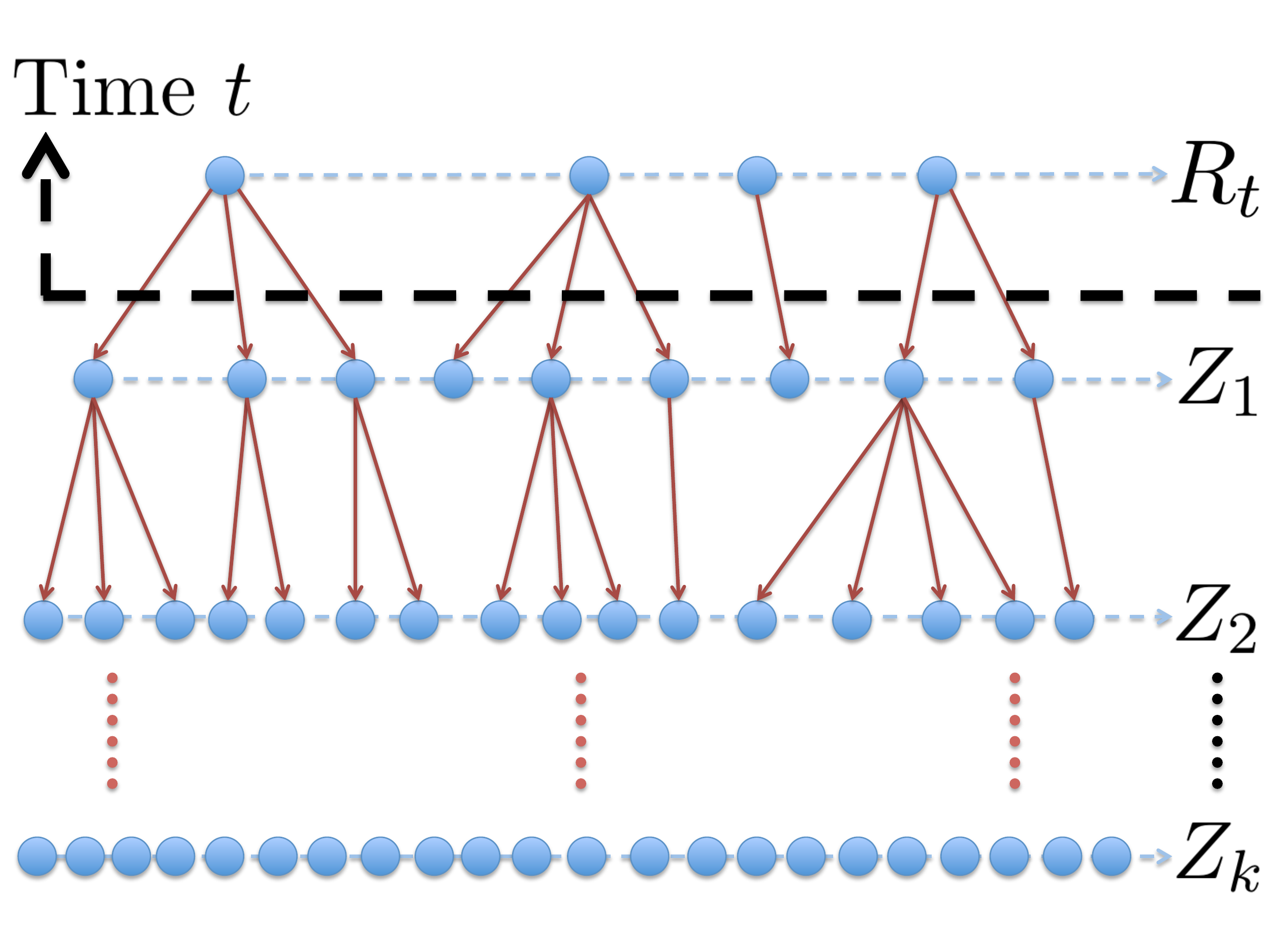}
\caption{\label{fig:descendants} An illustration of the information
  diffusion tree. We observe the cascade up to time $t$ (denoted by a dashed line)
  and the question is how the cascade tree is going to grow in the future.
  We define variables $Z_k$ which denote the number of reshares
  caused by the $k^\textit{th}$ generation descendants. Using variables $Z_k$
  the final reshare count $R_{\infty}$ can then be simply computer as $R_t + \sum\limits_{k = 1}^{\infty} Z_k$.}
\end{figure}

Given $Z_{1}$, the sequence of random variables $Z_{k}$ defines a
Galton-Watson tree with the \emph{offspring expectation} $\mu =
n_{*}p$ \cite{durrett2010probability}. Here, $\mu$ denotes the
expected number of reshares that the post gets.
Using a standard branching process result, we have $Z_{i}/\mu^i$ is a
martingale. Therefore, $\forall k>1$, $\E \left [Z_{k+1} | Z_k \right ]  = \mu \ Z_k$,
and,
\[
\E\left[\sum_{k=1}^\infty Z_k \middle| Z_1\right] = \frac{Z_1}{(1 - \mu)}= \frac{Z_1}{(1 - n_* p)}\, .
\]
Hence, we obtain
$$\E[R_\infty | \F_t] = R_t
+ \E\left[\sum_{k=1}^{\infty} Z_{k}\right] = R_t +
\frac{\E[Z_1]}{(1 - n_{*} p)},$$ which ends up being the right hand
side in Eq.~\eqref{eq:predict} because $\E[Z_1] = p(N_t - N^e_t)$ by
the definition of $Z_1$ and $N^e_t$.

Next, consider the case where $p = \p_t \geq 1/n_*$. In this regime,
the point process is supercritical and stays explosive. In terms of the
Galton-Watson tree discussed above, the offspring expectation $\mu =
n* p \ge 1$, so $\E[Z_{k+1}] \ge \E[Z_k] \ge \dotsb \ge
\E[Z_1]$. Therefore the total future reshares
$\sum_{k=1}^{\infty}Z_k$ has infinite expectation and the final
reshare count cannot be reliably predicted.
\end{proof}



Note that Prop.\ \ref{prop:expectation} assumes that the post
infectiousness remains constant in the future ($p_s = p_t$ for $s \ge
t$), which could be unrealistic for some information cascades. We correct this by changing the prediction formula in Eq.~\eqref{eq:predict} by adding two
scaling constants $\alpha_t, \gamma_t$ that adjust the final prediction:
\begin{equation}
  \label{eq:new-predict}
\hat{R}_{\infty}(t) = R_t + \alpha_t \frac{\hat{p}_t (N_t - N^e_t)}{1
  - \gamma_t \hat{p}_t n_{*}},~0<\alpha_t,\gamma_t<1 \, .
\end{equation}
We introduce these correction factors based on the
following intuition. We expect $\alpha_t$ to decrease over time $t$ so it scales
down the estimated infectiousness in the future, which accounts for
the post getting stale and outdated. Similarly, $\gamma_t$ accounts
for the overlap in the neighborhoods of reposters' followers. Over
time as the post spreads farther in the network, we expect $\gamma_t$
to increase as more nodes get exposed multiple times, which means the
arrival rate of new nodes (previously unexposed nodes) decreases over time.

We use the same values of $\alpha_t$ and $\gamma_t$ for all posts but allow them to vary over time.
The values of $\alpha_t$ and $\gamma_t$ are selected to minimized
median Absolute Percentage Error (refer to Section
\ref{sec:error-metrics} for definition) on a training data set. As
described in Section~\ref{sec:parameter-estimation}, we find
$\alpha_t$ is more important than $\gamma_t$ in
practice. 

\hide{
There are several reasons for using such model calibration procedure:
First, in decision theory and statistics, Stein's paradox shows that
  shrinking the estimated means $\hat{R}_{\infty}(t)$ is more accurate
  than using them directly, especially if we need to make a large number
  of predictions.
Second, we expect $p_t$ might decrease as time $t$ goes by (\eg, Figure
  \ref{fig:breakout}), so we may
  overestimate the size of first generation $\mathrm{E}[Z_{1}]$ and
  the offspring expectation $\mu$. To compensate for this, we discount them by
multiplying $\alpha_t$ and $\gamma_t$.
And third, the estimated memory kernel $\phi(s)$ and thus $N^e_t$
can be inaccurate. This error can also be corrected by introducing
tuning parameters.
}

\subsection{The SEISMIC algorithm}
\label{sec:algorithm}

Last, we put together all the components described so far and synthesize them in the \sem algorithm.
The \sem algorithm for predicting $\hat{R}_{\infty}(t)$ is described in Algorithm \ref{alg:2}, which uses the algorithm for computing $\hat{p}_t$ (Algorithm~\ref{alg:1}) as a subroutine.
These algorithms are based on Eqs. \eqref{eq:mle-kernel} and
\eqref{eq:new-predict}. We assume parameters
$K_t(s)$, $\alpha_t$, $\gamma_t$, $n_{*}$ are given \emph{a priori} or estimated from the data.

\begin{algorithm}[t]
  \caption{\sem: Predict final cascade size}
  \label{alg:2}
  \begin{algorithmic}
    \State  \textbf{Purpose:} {For a given post  at time $t$, predict its final reshare count}
    \State \textbf{Input:} {Post resharing information:} $t_i~\text{and}~n_i~\text{for}~i = 0, \dotsc,R_t$.

    \State \textbf{Algorithm:}
    \State
    $N_{t}$ = 0, $N_{t}^{e}$ = 0
    \For {$i =  0, \dotsc, R_t$}
    \State $N_{t}$ += $n_i$
    \State $N_{t}^{e}$ += $n_i\int_{t_i}^{t} \phi(s - t_i)ds$ \hfill(Sec.~\ref{sec:human-reaction-time})
    \EndFor
    \State $\hat{R}_{\infty}(t) = R_t + \alpha_t \hat{p}_t (N_t - N^e_t) / (1 - \gamma_t \hat{p}_t n_{*})$ \hfill(Alg.~\ref{alg:1})
    \State \textbf{Deliver:} $\hat{R}_{\infty}(t)$
  \end{algorithmic}
\end{algorithm}

\begin{algorithm}[t]
  \caption{Compute real-time infectiousness $\hat{p}(t)$}
  \label{alg:1}
  \begin{algorithmic}
    \State  \textbf{Purpose:} {For a given post $w$, calculate infectiousness $p_t$ with information about $w$ prior to time $t$}
    \State \textbf{Input:} {Post resharing information:} $t_i~\text{and}~n_i~\text{for}~i = 0, \dotsc, R_t$.

    \State \textbf{Algorithm:}
    \State
    $\tilde{R}_t$ = 0, $\tilde{N}_t^{e}$ = 0
    \For {$i =  0, \dotsc, R_t$}
    \State $\tilde{R}_t$ += $K_t(t - t_i)$
    \EndFor
    \For {$i =  0, \dotsc, R_t$}
    \State $\tilde{N}_t^{e}$ += $n_i\int_{t_i}^{t} K_t(t - s)\phi(s - t_i)ds$
    \hfill(Sec.~\ref{sec:estimation-infectiousness})
    \EndFor

\State
$p_t = \tilde{R}_t/\tilde{N}_t^{e}$
    \State \textbf{Deliver:} $p_t$
  \end{algorithmic}
\end{algorithm}

\xhdr{Computational complexity of \sem}
For any choice of $\phi(s)$ and $K_t(s)$, the computational cost of \sem is
$O(R_{t})$ for both calculating $\hat{p}_t$ and predicting $\hat{R}_{\infty}(t)$.
Of course, the actual computing time depends heavily on the
integration $\int_{t_i}^{t} K_t(t - s)\phi(s - t_i)ds$
and  $\int_{t_i}^{t} \phi(s - t_i)ds$. However, the overall computational cost of \sem is {\em linear} in the observed number of reshares $R_t$ of a given post by time $t$.

The linear time complexity is in part also due to the shape of our memory kernel. In Section~\ref{sec:parameter-estimation} we will estimate the memory kernel $\phi(s)$ for Twitter to have  the following form (for some $s_0 > 0$):
\begin{align}
  \label{eq:memory-kernel}
  \phi(s) =
  \begin{cases}
  c  &\text{if $0 < s \le s_0$}, \\
  c(s/s_0)^{-(1+\theta)} & \text{if $s > s_{0}$}.
  \end{cases}
\end{align}
This means that with the memory kernel $\phi(s)$ in Eq.~\eqref{eq:memory-kernel} and triangular weighting kernel $K_t(s)$ in Eq.~\eqref{eq:weighting-kernel}, both integrals can be evaluated in closed form because they are piece-wise polynomials (polynomial with possibly non-integer exponents), which greatly decreases computational cost of \sem.


%% file: Experiments.tex
In this section, we describe the Twitter data set, our parameter estimation procedure, and compare the performance of \sem to state-of-the-art approaches.

\subsection{Data description and data processing}
Our data is the complete set of over 3.2 billion tweets and retweets on Twitter from October 7
to November 7, 2011. For each retweet, the dataset includes tweet id,
posting time, retweet time, and the number of followers of the poster/retweeter.
Note, the data set lacks Twitter network information. The only piece of network information available to us is the number of followers of a node.

We focus on a subset of reasonably popular tweets with at least 50 retweets, so that our model enables the prediction as soon as sufficient number of retweets occur. Note, that if multiple Twitter users independently post the same tweet, which then gets retweeted, each original posting creates its own independent cascade. All in all there are 166,076 tweets
satisfying this criterion in the first 15 days. We form the training set
using the tweets from the first 7 days and the test set using the tweets from the next 8 days. We use the remaining 14 days for the retweet cascades to unfold and evolve.
For a particular retweet cascade, we obtain all the retweets posted within 14 days of the original post time, \ie, we approximate $R_{\infty}$ by $R_{\text{14 days}}$.
We estimate parameters $\phi(s), \alpha_t, \gamma_t
\text{ and } n_{*}$ with the training set, and evaluate the performance of
the estimator $\hat{R}_{\infty}$ on the test set. For the tweets in
our training set, $R_{\text{14 days}}$ has mean 209.8 and median
110. The temporal evolution of mean and median of $R_t$ are also shown in Figure \ref{fig:rtcum}.


\begin{figure}[t]
  \centering
  \includegraphics[width = \columnwidth]{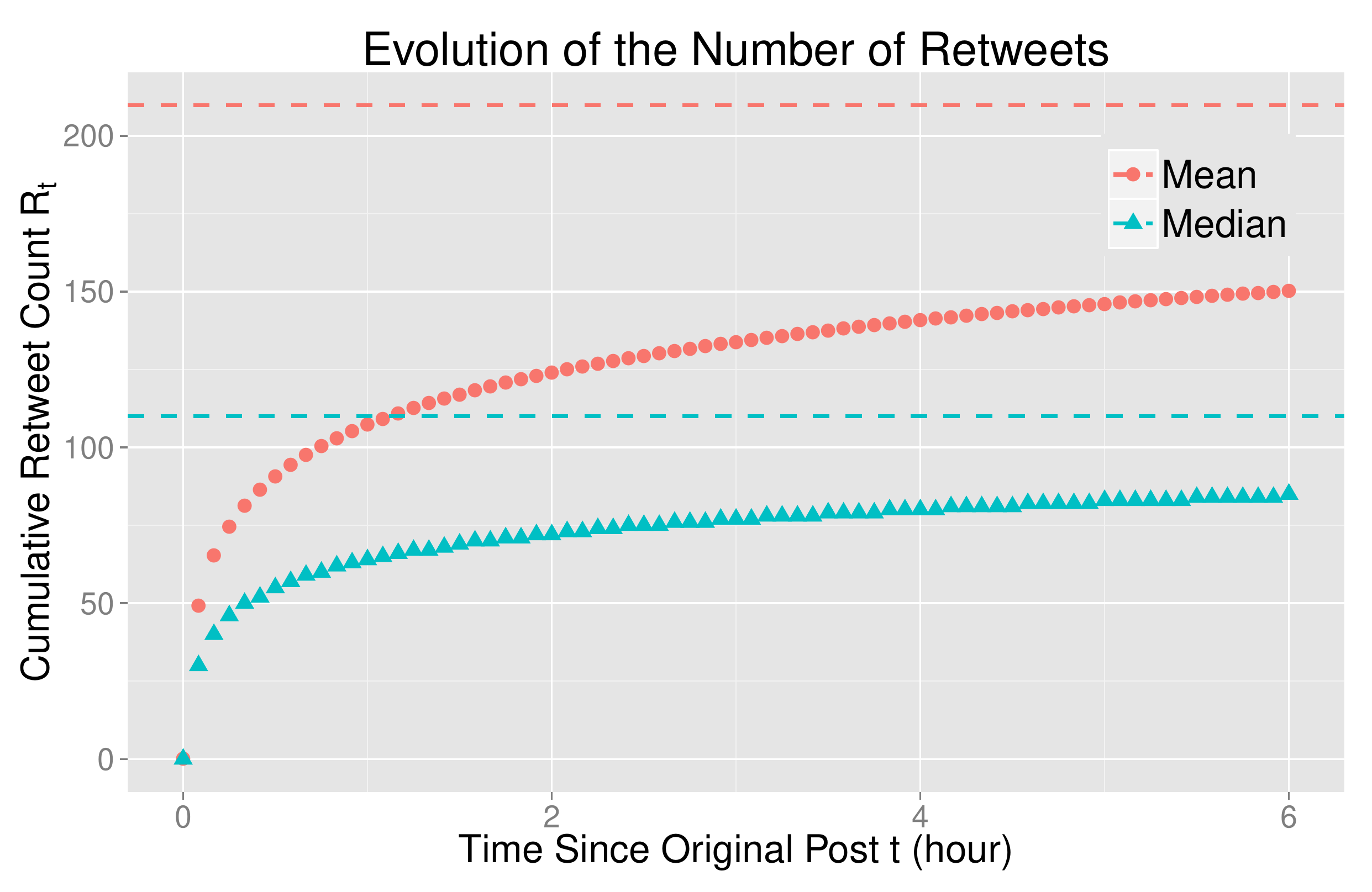}
  \caption{Convergence of the mean and media cumulative retweet count $R_{t}$
  as a function of time.The horizontal lines correspond to mean and median final
  retweet count $R_{14\,\text{days}}$. On average, a tweet receives 75\% of its retweets in the first 6 hours.}
  \label{fig:rtcum}
\end{figure}

\subsection{SEISMIC parameter estimation}
\label{sec:parameter-estimation}
First we describe how to fit the memory kernel $\phi(s)$
(Section~\ref{sec:human-reaction-time}). We
carefully choose 15 tweets in the training set and use the
distribution of all their retweet times as our $\phi(s)$ (Figure
\ref{fig:pdf}). The histograms of the 15 sequences of retweet times all
display a clear shape of subcritical decay. Moreover, all the original
posters have an overwhelming number of followers. Therefore, most of the
retweets, if not all, should come from the immediate followers of the
original poster. Consequently, the distribution of human reaction time
can be well approximated by that of the retweet times of these 15
tweets. The estimation of $\phi(s)$ can be further improved if the
network structure is available.

The observed reaction time distribution (Figure \ref{fig:pdf}) suggests a form of
Eq.~\eqref{eq:memory-kernel} for the memory kernel: constant in the
first 5 minutes, followed by a power-law decay. After setting the constant
period $s_0$ to 5 minutes, we estimate power law decay parameter
$\theta =$ 0.242 with the complimentary cumulative distribution function
(ccdf), and chose $c = 6.27 \times 10^{-4}$ to make
$\int_{0}^{\infty} \phi(s) ds = 1$. The memory kernel is a network wide
parameter and only needs to be estimated once. The fitted memory kernel is
plotted in Figure \ref{fig:pdf}.
\begin{figure}[t]
  \centering
  \includegraphics[width = \columnwidth]{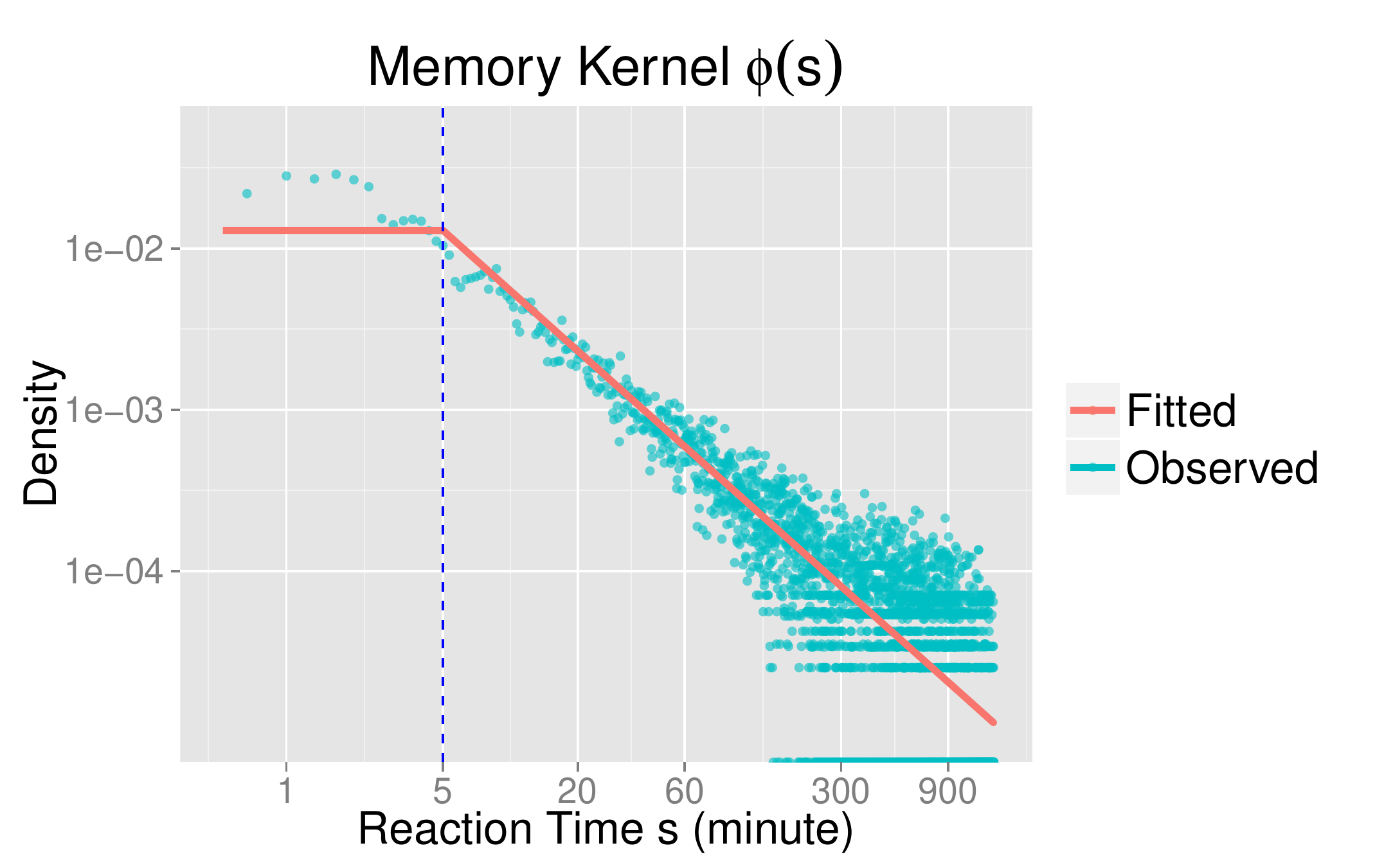}
  \caption{Reaction time distribution and the estimated memory kernel $\phi(s)$.
  The reaction time is plotted on logarithmic axes, hence the linear trend suggests
    a power law decay.}
  \label{fig:pdf}
\end{figure}

\hide{ 
With the simplified
assumptions on the contagiousness of the 15 selected
tweets, we run the risk of overestimating the tail of $\phi(s)$, which
results in an underestimation of the ``effective'' degree
$N^e_t$. We compensate for the overestimation of $\phi(t)$
by adding parameters $\alpha_t$ and $\gamma_t$ in the final
prediction (See Eq.~\eqref{eq:predict}). As discussed earlier in
Section~\ref{sec:predict-influence}, the $\alpha_t$ and $\gamma_t$
also help correct the bias from assuming constant infectiousness parameter.
} 



Last, we briefly comment on the correction factors $\alpha_t$ and $\gamma_t$ introduced in Eq.~\eqref{eq:predict}.
We use the same values of $\alpha_t$ and $\gamma_t$ for all tweets. Notice that $\gamma_t$ and $n_{*}$ only affect the predictions through their product $\gamma_t
n_*$. Overall, we find the value of $\gamma_t n_*$ has little effect on the
performance of our algorithm. In our experiments we simply set $\gamma_t n_* =$ 20 for all $t$.
We choose the value of
$\alpha_t$ such that it minimizes the training median Absolute Percentage Error (Section \ref{sec:error-metrics}). We report values of $\alpha_t$ in Table \ref{tab:alpha}. $\alpha_t$ has a particularly
small value at $t =$ 5 minutes, which may be a result of the overestimation
of $p_t$, when the triangular kernel has not moved away from the unstable beginning period.
After that $\alpha_t$ begins a slow and consistent decay to account for the fact that information is getting increasingly stale and outdated over time.

With all the estimated parameters in \sem, we are ready to apply
it (Algorithms \ref{alg:2} and
\ref{alg:1}) to the Twitter dataset. For a given tweet $w$ and every 5 minute interval $t$, we output our estimate $\hat{R}_\infty(t,w)$ of the tweet's final retweet count $R_\infty(w)$.


\begin{table}[t]
\centering
\begin{tabular}{c|rrrrr}
time (minute) & 5 & 10 & 15 & 20 & 30 \\
  \hline
  $\alpha$ & 0.389 & 0.803 & 0.772 & 0.709 & 0.680  \\
   \hline\hline
time (minute) & 60 & 120 & 180 & 240 & 360 \\
\hline
$\alpha$ & 0.562 & 0.454 &
  0.378 & 0.352 & 0.326 \\
\end{tabular}
\caption{Values $\alpha_t$ used in Algorithm \ref{alg:2}.}
\label{tab:alpha}
\end{table}

\subsection{Baselines for comparison}
\label{sec:baseline-estimators}

We consider four different prediction methods for comparison. The
first two are regression based and the next two are point process based.

\begin{itemize}[nolistsep]
\item{\bf Linear regression (LR) \cite{Szabo2010}:} The model can be defined as
$$\log R_{\infty} = \alpha_t + \log R_t +
  \epsilon,$$ where $\epsilon$ denotes the Gaussian noise. This is
  also the second baseline estimator used in \cite{Zaman2014}. 
  Notice that all the
  tweets receive the same multiplicative constant for a given time.
\item{\bf Linear regression with degree (LR-D) \cite{Szabo2010}:} This model can be written as
$$\log R_{\infty} = \alpha_t + \beta_{1,t} \log R_t +
  \beta_{2,t} \log N_t + \beta_{3,t} \log n_0 + \epsilon$$ where
  $\epsilon$ denotes, as before, the Gaussian noise. LR-D is
  more flexible than LR, since it allows $\log R_t$ to have a slope
  not equal to $1$ and uses additional features.
\item{\bf Dynamic Poisson Model (DPM) \cite{Agarwal2009, Crane2008}:}
It models the retweet times $\{t_k \}$
as a point process with rate
\[
\lambda_t = \lambda_{t_{\textrm{peak}}} (t -
  t_{\textrm{peak}})^{\gamma}
\]
where $t_{\textrm{peak}} = \arg\max_{s
  < t} \lambda_s$. The power-law parameter $\gamma$ is estimated
separately for each tweet. To discretize the model, we bin the retweet
times into $b=$ 5 minute intervals. Note that when $\gamma > -1$, the integral
$\int_{t_{\mathrm{peak}} + b}^{\infty} \lambda_t dt$ is
infinite. In such cases, we move $t_{\mathrm{peak}}$ forward to
the second maximum bin.
\item{\bf Reinforced Poisson Model (RPM) \cite{gao2015modeling}:} This
  recently published state-of-the-art approach models the reshare rate as
\[
\lambda_t = c f_{\gamma}(t)r_{\alpha}(R_t)
\]
where parameter $c$ measures the attractiveness of the message, $f_{\gamma}(t)
\propto t^{-\gamma} (\gamma>0)$ models the
aging effect, and $r_{\alpha}(R_t) (\alpha > 0)$ is the reinforcement
function which depicts the ``rich get richer'' phenomenon.
Given a particular tweet, the parameters $c, \gamma, \alpha$ are found by
maximizing the likelihood function, where the optimal values are projected to
their feasible sets whenever they are out of range.
\end{itemize}

\subsection{Evaluation metrics}
\label{sec:error-metrics}

For a particular tweet, suppose that the prediction for $R_{\infty}$ at
time $t$ is denoted by $\hat{R}_{\infty}(t)$. We use the following evaluation metrics in our experiment:
\begin{itemize}[nolistsep]
\item{\bf Absolute Percentage Error (APE):} For a given tweet $w$ and a prediction time $t$, the APE metric is defined as,
\eqn{
\text{APE}(w,t)=\frac{|\hat{R}_{\infty}(w,t) -
    R_{\infty}(w)|}{R_{\infty}(w)}.
}
     When the APE metric is used for evaluation purposes, various
     quantiles of APE over the tweets (all possible $w$) in the test
     dataset will be reported at each time $t$.
\item{\bf Kendall-$\tau$ Rank Correlation:} This is a measure of rank correlation
  \cite{kendall1938new}, which computes the
  correlation between the ranks of $\hat{R}_{\infty}(t)$ and
  $R_{\infty}$ for all test tweets. This metric is generally more robust than Pearson's
  correlation of values of $\hat{R}_{\infty}(t)$ and
  $R_{\infty}$.
  A high value of rank correlation means the predicted and the final retweet counts are strongly correlated.
\item{\bf Breakout Tweet Coverage:} We create a ground-truth list of top-$k$ tweets with the highest final retweet count. We refer to these tweets as ``breakout'' tweets.
  Using our model we can also produce a top-$k$ list based on the predicted final retweet count. We evaluate the methods by quantifying how well the predicted top-$k$ list covers the ground-truth top-$k$ list. We give additional details in Section~\ref{sec:identify-breakout-tweets}.
\end{itemize}



\subsection{Experimental results}
\label{sec:results}

In this section we evaluate the performance of \sem and the four competitors described in Section
\ref{sec:baseline-estimators}. All the methods start making predictions as
soon as a given tweet gets retweeted 50 times.

\subsubsection{SEISMIC model validation}
\label{sec:model-validation}

First, we empirically validate \sem. In Proposition
\ref{prop:expectation}, we obtain a formula for the expected number of final
retweets in terms of the infectiousness parameter $p_t$. Our goal here is to show that Proposition
\ref{prop:expectation} provides an unbiased estimate of the true final retweet count. We proceed as follows.
We use \sem to make a prediction after observing each tweet for 1 hour and then plot the prediction against the true final number of retweets. If \sem gives an unbiased estimate, then we expect a diagonal curve $y=x$, that is, the expected predicted $\hat{R}_{\infty}$ matches the true expected $R_\infty$.

Figure~\ref{fig:valid} shows that the empirical average almost perfectly coincides with \sem's
prediction. This suggests that the \sem estimator in Eq.~\eqref{eq:predict} is unbiased and
we can safely use it to predict the expected final number of retweets.
However, as mentioned earlier, in practice one often wants to shrink the
prediction in order to stabilize the estimator and achieve better
overall performance.
Therefore, we use the calibrated prediction formula
Eq.~\eqref{eq:new-predict} for the rest of the experiments.

\begin{figure}[t]
  \centering
  \includegraphics[width = \columnwidth]{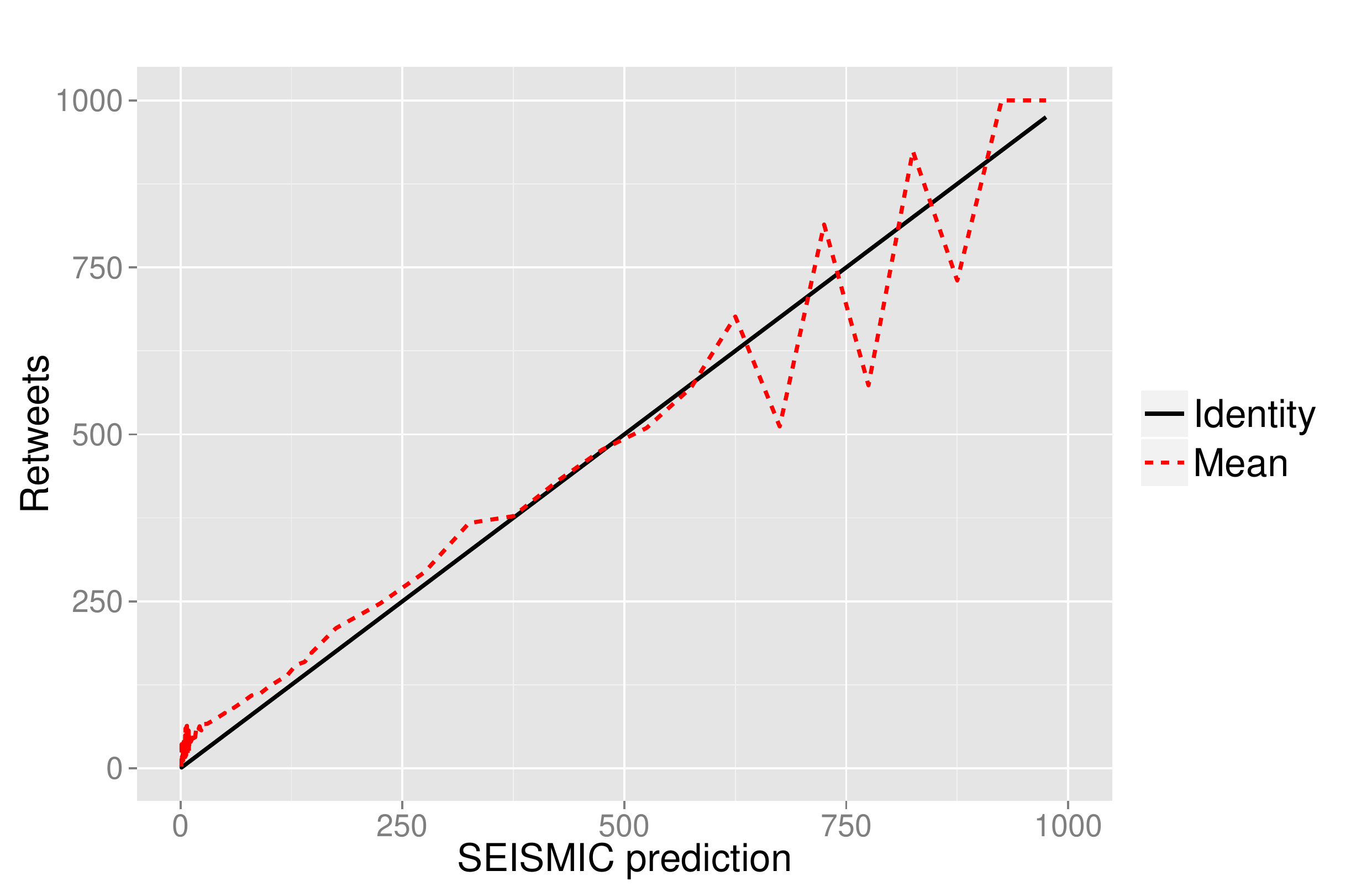}
  \caption{Predicted final retweet counts nicely follow the ground-truth retweet counts,
  which suggests \sem provides an unbiased estimate of the final
  retweet count. The dashed red curve is obtained by binning the
  tweets according to the prediction and then computing the average
  number of retweets in each bin.}
  \label{fig:valid}
\end{figure}

\subsubsection{Predicting final retweet count}

We run our \sem method for each tweet and compute the Absolute
Percentage Error (APE) as a function of time. We plot the quantiles of the distribution of APE of \sem in Figure~\ref{fig:ape-ours}.
After observing the cascade for 10 minutes ($t=$10 min), the 95th,
75th, and 50th percentiles of APE are less than 71\% , 44\%, and 25\%, respectively. This means that after 10 minutes, average error is less than 25\% for 50\% of the tweets and less than 71\% for 95\% of the tweets. After 1 hour the error gets even lower---APE for 95\%, 75\% and 50\% of the tweets drops to 62\%, 30\% and 15\%, respectively.

\begin{figure}[t]
  \centering
  \includegraphics[width = \columnwidth]{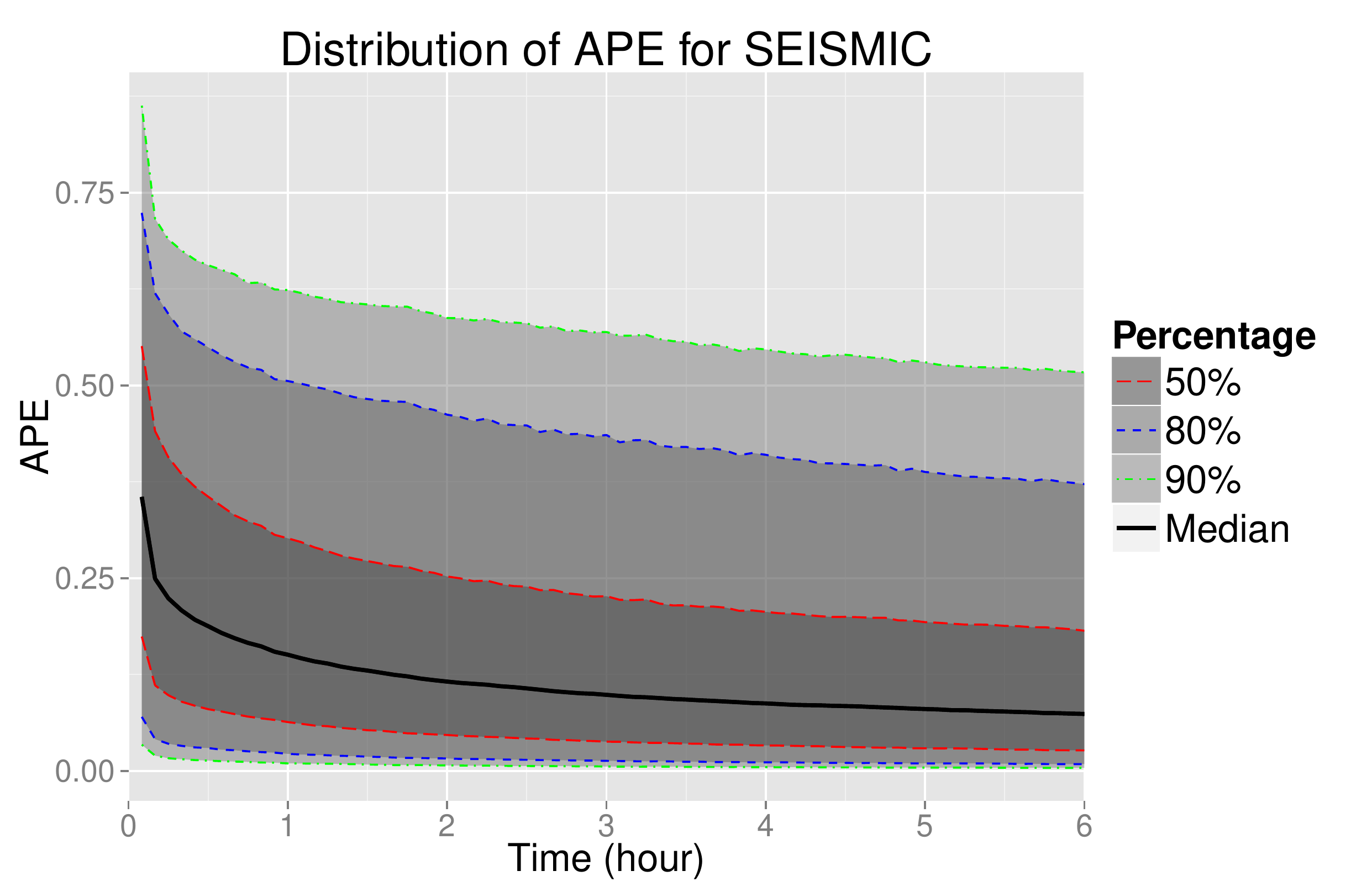}
  \caption{Absolute Percentage Error (APE) of \sem on
    the test set. We plot the median and the middle 50th, 80th, 90th
    percentiles of the distribution of APE across the tweets.}
  \label{fig:ape-ours}
\end{figure}

The proposed method, \sem, demonstrates a clear improvement over the
baselines and the state-of-the-art as
shown in Figures \ref{fig:med-cor} and \ref{fig:zoom-in}. The left panels of Figures
\ref{fig:med-cor} and \ref{fig:zoom-in} show the median APE of different methods over time as more and more of the retweet cascade gets revealed. The LR and LR-D baselines
have very similar performances, indicating the additional features used
by LR-D are not very informative. DPM performs poorly across the
entire tweet lifetime, while the other point process approach RPM is worse than LR
and LR-D in the early period but becomes better after about 2 hours.
All in all, in terms of median APE score \sem is about $30\%$ more accurate than all the competitors
across the entire twee lifetime.

Similarly, the right panels of Figures \ref{fig:med-cor} and \ref{fig:zoom-in}
show the Kendall-$\tau$ rank correlation between the predicted ranking
of top most retweeted tweets and the ground-truth ranking of
tweets. Again \sem is giving much more accurate rankings than other methods.


\begin{figure}[t]
  \centering
  \includegraphics[width = \columnwidth]{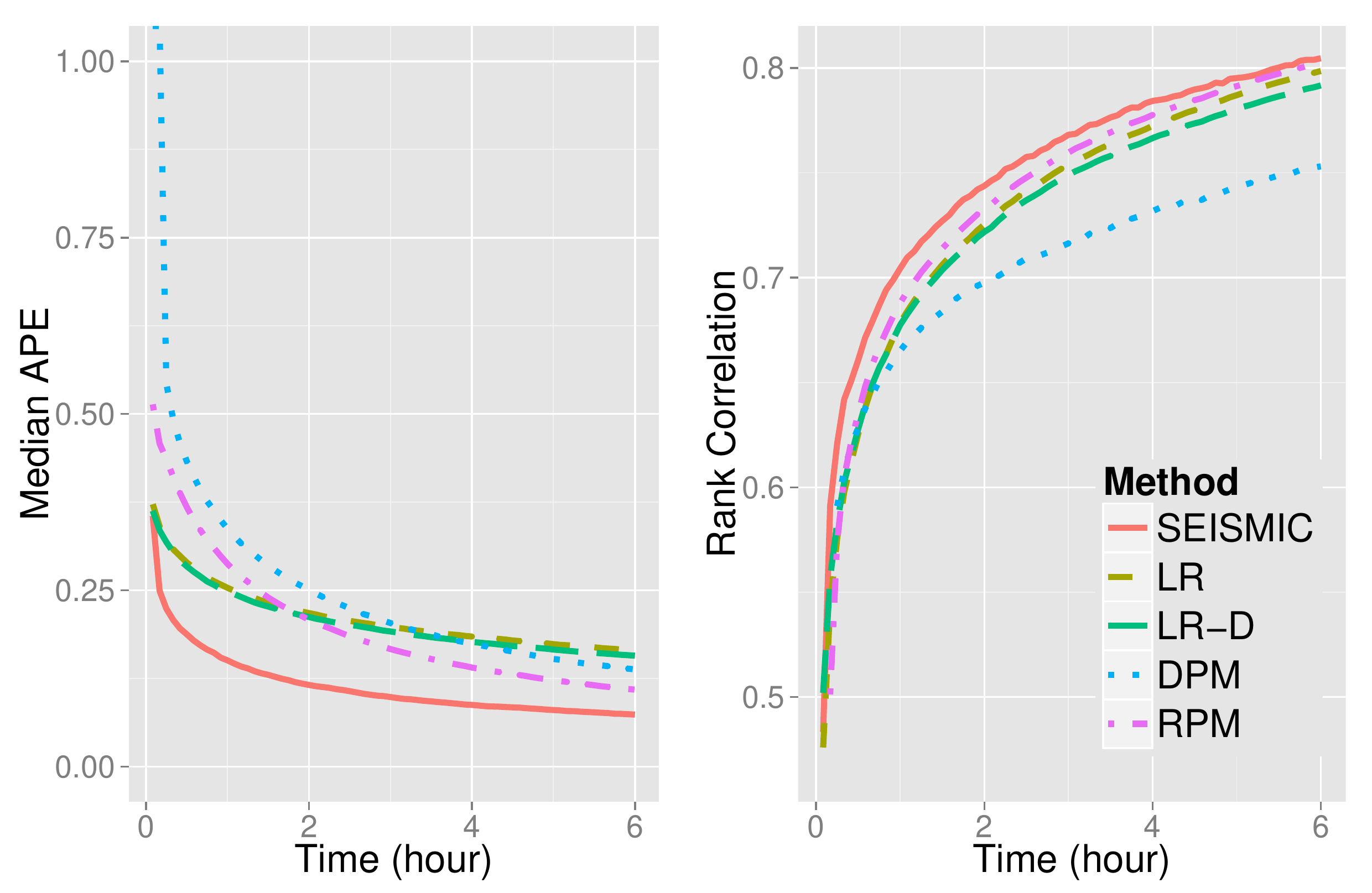}
  \caption{Median Absolute Percentage Error (APE) and Kendall's Rank Correlation
    of \sem and the baselines as a function of time. \sem consistently gives best performance.}
  \label{fig:med-cor}
\end{figure}

\begin{figure}[t]
  \centering
  \includegraphics[width = \columnwidth]{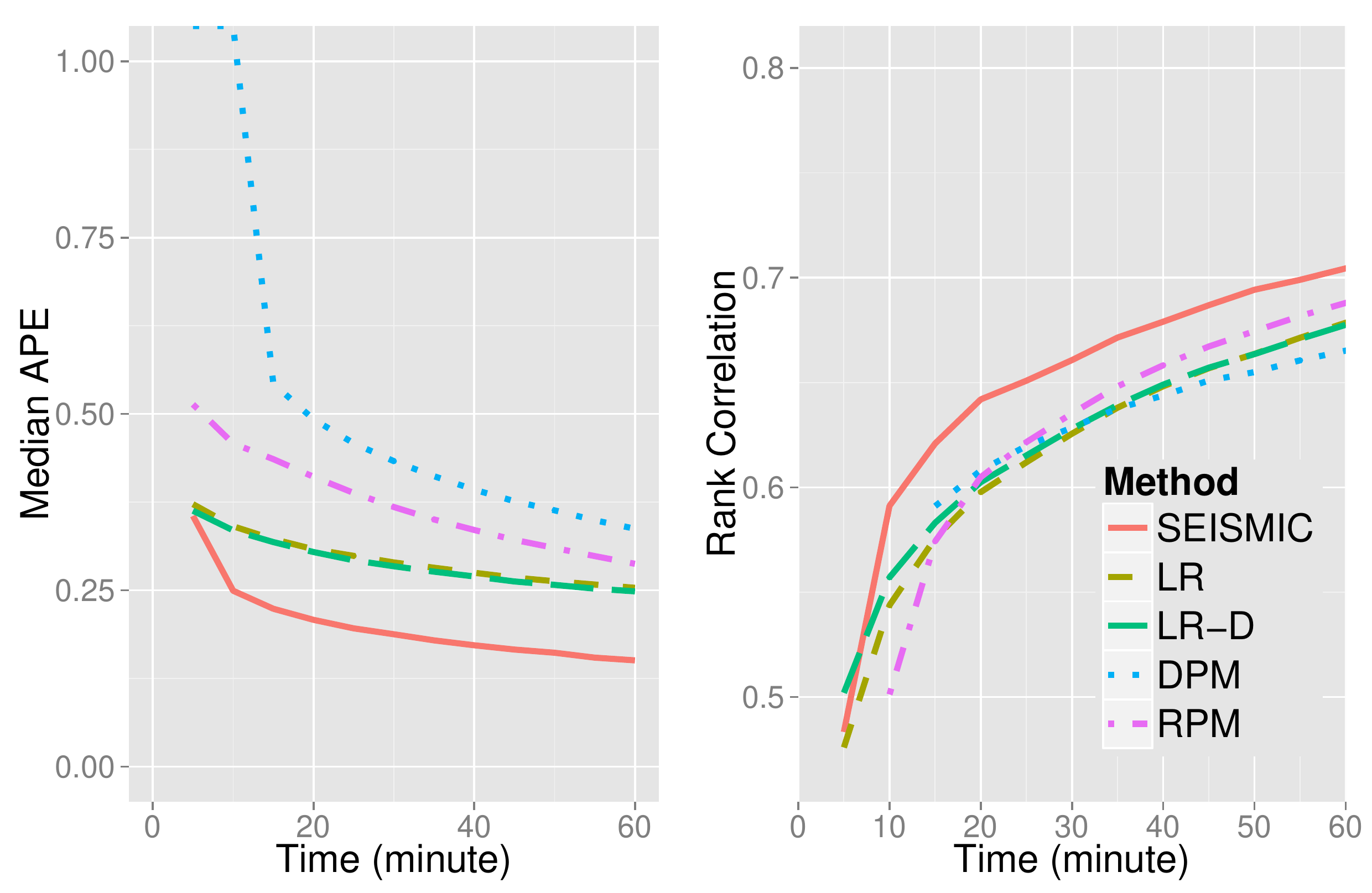}
  \caption{Zoom-in of Figure \ref{fig:med-cor}: Median APE and Rank
    Correlation for the first 60 minutes after the tweet was
    posted. \sem performs especially well compared to the baselines
    early in the tweet's lifetime.}
  \label{fig:zoom-in}
\end{figure}


\subsubsection{Identifying breakout tweets}
\label{sec:identify-breakout-tweets}


Can we identify a breakout tweet before it receives most of its
retweets? This question arises from various applications like trend
forecasting or rumor detection. The goal of this prediction task is to as early as possible identify ``breakout'' tweets, which have the highest final retweet count.
We quantify the performances of different models in detecting breakout tweets by using models' predictions of tweets' final retweet counts.

First, we form a ground-truth set $L_{M}^{*}$  of size $M$. The set
$L_{M}^{*}$ contains top-$M$ tweets with the highest final retweet
counts. Then with each of the prediction methods, we produce a
sequence of size $m$ lists, $\hat{L}_{m}(t)$. At each time $t$ the list $\hat{L}_{m}(t)$ contains the top-$m$ tweets with the highest predicted retweet counts at time $t$.

As described in Section~\ref{sec:error-metrics}, we then compare each $\hat{L}_m(t)$ with $L_M^{*}$, and calculate the {\em Breakout Tweet Coverage}, which is defined as the proportion of tweets in $L_M^{*}$ covered by $\hat{L}_m(t)$. 

Fig.~\ref{fig:rank} shows the performance of \sem in detecting top 100 most retweeted tweets ($L_{100}^{*}$) as a function of time.  \sem is able to cover 82 tweets in the
first 1 hour and 93 tweets in the first $6$ hours.

The fifth most retweeted tweet in this plot is actually the tweet we
showed earlier in Figure~\ref{fig:breakout}. We observe that \sem
detects this tweet 30 minutes after it has been posted, while LR and
LR-D both take more than an hour. DPM fails to detect this breakout for the first 6 hours (plots not show for brevity).


\begin{figure}[t]
  \centering
  \includegraphics[width = \columnwidth]{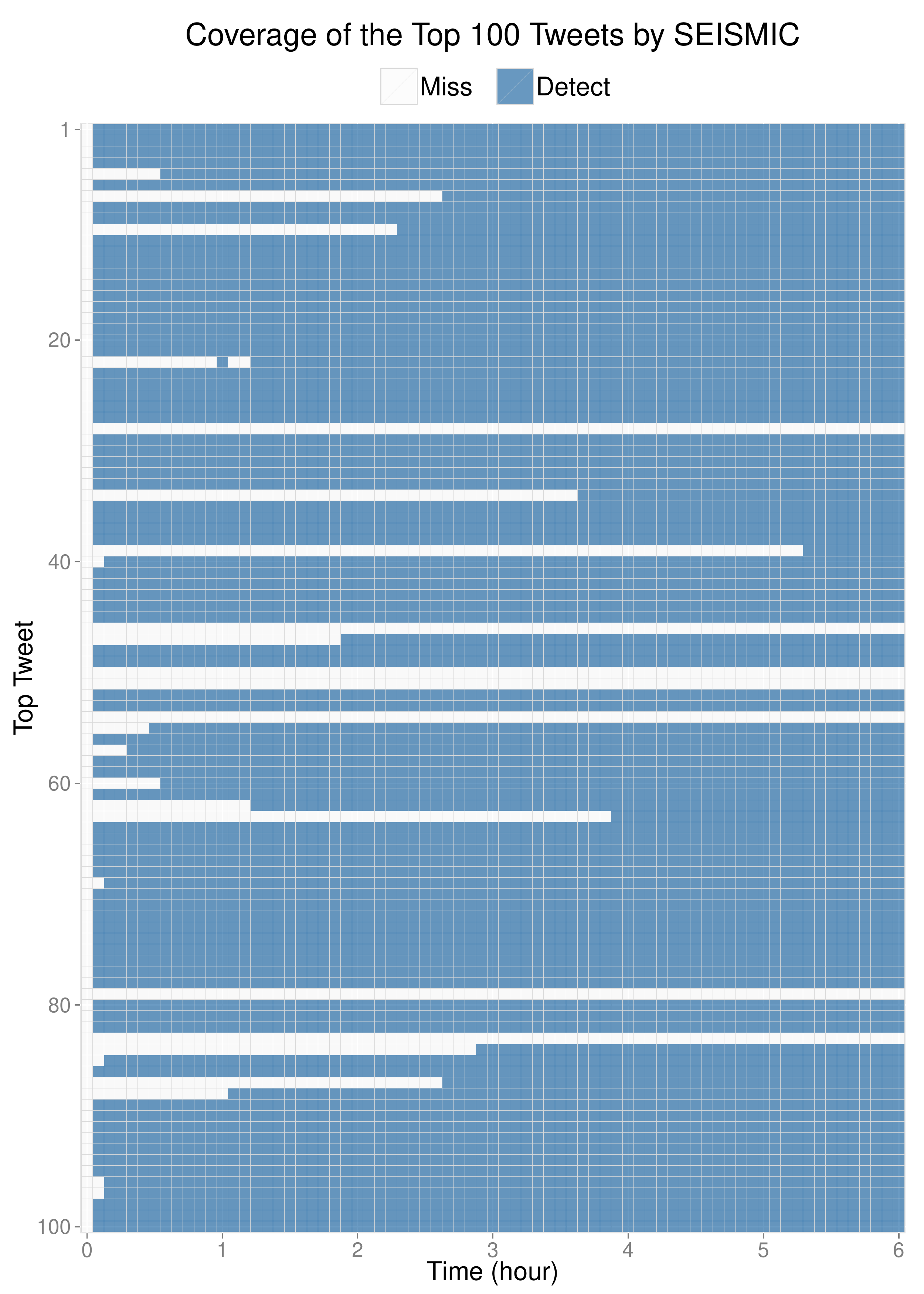}
  \caption{Coverage of top 100 most retweeted tweets.
    Each row represents a tweet.
    White blocks indicate that a given tweet was not covered by \sem's
    predicted list of top-500 tweets at time $t$, and blue indicates successful coverage.
    }
  \label{fig:rank}
\end{figure}

To compare \sem with other methods, we keep the size of the
predicted lists to be $m =$ 500, and use a larger target list
$L_{500}^{*}$, which is a more difficult task than finding $L_{100}^{*}$.
Figure \ref{fig:coverage} compares the coverage of
different methods against the proportion of retweets seen. After
seeing 20\% of the retweets, \sem covers 65\%
of the shortlist, while LR-D and LR both
cover only 50\%. In general, the dynamic Poisson model fails to
provide accurate predictions and breakout identifications.

Overall, \sem allows for effective detection of breakout tweets. For instance, after seeing around 25\% of the total number of retweets of a given tweet (in other words, after observing a tweet for around 5 minutes), \sem can identify 60\% of the top-100 tweets according to the final retweet counts.

\begin{figure}[t]
  \centering
  \includegraphics[width = \columnwidth]{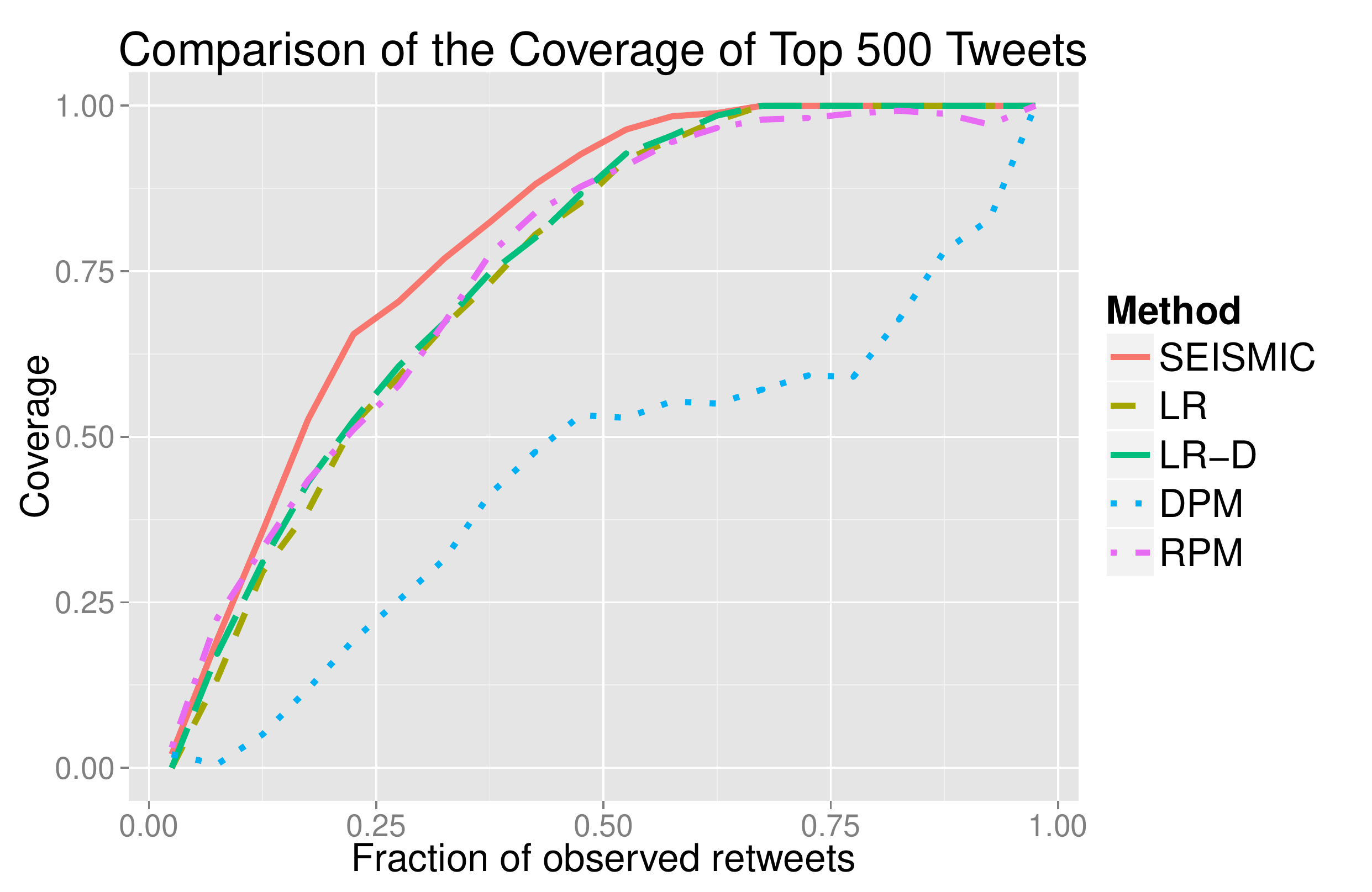}
  \caption{Coverage of top 500 tweets ($L_{500}^{*}$) by various methods.
  \sem exhibits clear improvement over all
  methods after about 10\% of retweets are observed.
  All methods except DPM achieve perfect coverage after 65\% of retweets
  are observed.}
  \label{fig:coverage}
\end{figure}

\subsection{Discussion of model robustness}
\label{sec:comp-with-state}

\sem demonstrates better robustness than the other two point process
based methods --- DPM and RPM. While \sem is not able to make a prediction
for tweets that are in the supercritical state, DPM and RPM are unable
to make predictions when the decay parameter is outside the feasible set ($\gamma < - 1$ for DPM and $\gamma < 0$ or
$\alpha < 0$ for RPM). For example, in
Figure~\ref{fig:breakout}, \sem characterizes the tweet as
supercritical for the first 70min, DPM fails to make a prediction for the first
6 hours and RPM is only able to make a prediction from 30 to 80 minute.

All in all, we find that tweets are in the supercritical regime for
only a very short time and \sem is able to make predictions for most
of the tweets in most of the time. We find that on average, \sem is not able to make a prediction for 1.80\% of the tweets after observing them for 15 minutes. In other words, after 15 minutes, 1.80\% of the tweets are still in the supercritical regime (over all the tweets with at least 50 retweets). This number drops to 1.29\% (0.67\%) after 1 hour (6 hours). As a point of comparison we also note that other methods are not able to make predictions for a much larger fraction of tweets: DPM fails to make a prediction for 6.77\%, 5.79\% and 1.45\%, and RPM fails for 3.45\%, 5.69\% and 15.43\% of the tweets after 15min, 1h, and 6h.



Our \sem method is also significantly faster than the RPM model
\cite{shen2014modeling}, which requires to solve a nonlinear
optimization problem every time it predicts. In our implementation, the average running
time per tweet for predicting at every 5 minutes for 6 hours is 0.02s
for \sem and 3.6s for RPM. The reported running time includes both parameter learning and prediction.



%% file: Discussion.tex



\label{sec:concluding-remarks}


In this paper we propose \sem, a flexible framework for modeling information cascades and predicting the final size of an information cascade. Our contributions are as follows:
\begin{itemize}[nolistsep]
        \item We model the information cascades as self-exciting point
        processes on Galton-Watson trees. Our approach provides a theoretical framework for
        explaining temporal patterns of information cascades.
        \item \sem is both scalable and accurate. The model requires no feature engineering and scales linearly with the number of observed reshares of a given post. This provides a way to predict information spread for millions of posts in an online real-time setting.
        \item \sem brings extra flexibility to estimation and prediction tasks as it requires minimal knowledge about the information cascade as well as the underlying network structure.
\end{itemize}

There are many interesting venues for
future work and our proposed model can be extended in many different
directions. For example, if the network structure is available, one could replace the node degree $n_i$ by the number of newly exposed followers. If content-based features or features of the original post are available, one could develop a content-based prior of $p_t$ for each post. If temporal features such as the users' time zone are available, one could directly use them to modify the estimator $\hat{p}_t$. In this sense, the proposed model provides an extensible framework for predicting information cascades.

\sem is a statistically sound and scalable bottom-up model of information cascades that allows for predicting final cascade size as the cascade unfolds over the network. We hope that our framework will prove useful for developing richer understanding of cascading behaviors in online networks, paving ways towards better management of shared content. 